\documentclass[11pt]{article}
\usepackage[preprint]{acl}
\usepackage{times}
\usepackage{latexsym}
\usepackage[T1]{fontenc}
\usepackage[utf8]{inputenc}
\usepackage{microtype}
\usepackage{inconsolata}
\usepackage{graphicx}
\usepackage{amsmath}

\usepackage{amssymb}
\usepackage{multirow}
\usepackage{enumitem}
\usepackage{booktabs}
\usepackage{makecell}
\usepackage{array}
\usepackage{amsthm}
\usepackage[most]{tcolorbox}
\newtheorem{proposition}{Proposition}
\usepackage{listings}
\usepackage{needspace}
\newcommand{\eci}{\ensuremath{\mathrm{ECI}_{\mathrm{sem}}}}
\title{What Tracks Hard-Negative Transfer?\\ A Study of Lexical and Semantic Signals}

\author{
  Aarush Sinha$^{1}$ \quad
  Rahul Seetharaman$^{2}$ \quad
  Aman Bansal$^{2}$ \\
  $^{1}$University of Copenhagen, Denmark \\
  $^{2}$Independent Researcher, USA \\
  \texttt{\{aarush.sinha, rahulseetharaman, 18aman.bansal\}@gmail.com}
}

\begin{document}
\maketitle
\begin{abstract}
Which properties of hard-negative supervision track dense-retriever
transfer before training? We compare lexical noncoverage, frozen-encoder
similarity statistics, hardness proxies, and combinations of these signals
across BM25, dense, LLM-generated, and hybrid sources for MS MARCO training.
Evaluation uses three retriever backbones and eleven BEIR datasets. Mean
lexical residuality, an IDF-weighted measure of query-token noncoverage,
reproduces the aggregate transfer ordering in all six backbone and setting
blocks. The individual semantic statistics recover none of the complete
orderings, and combining them with the lexical signal does not improve the
observed rank agreement. After approximately matching BM25, dense, and LLM
diagnostic samples for lexical residuality, none of the tested semantic
statistics recovers the observed LLM $>$ Dense $>$ BM25 transfer ordering
across three backbones, three matching resolutions, and 100 repetitions.
Constructed controls further show that unrelated and false-negative-shaped
passages can exhibit favorable lexical and semantic statistics. The findings
identify a simple lexical association in this setting while distinguishing
source-level transfer agreement from the reliability of individual negatives.
\end{abstract}

\section{Introduction}
Dense text retrievers learn from contrasts between queries, labeled positive
passages, and negative passages
\cite{karpukhin-etal-2020-dense,xiong2021approximate,qu2021rocketqaoptimizedtrainingapproach,zhan2021optimizing,robinson2021contrastive}.
Choosing the negative source is consequential: lexical retrieval, dense
mining, and LLM generation produce different supervision, and evaluating
each source through a complete training and transfer pipeline is costly.
Training-free diagnostics could help prioritize these comparisons. More
fundamentally, understanding what those diagnostics measure can clarify
which properties of supervision accompany transfer across text retrieval
tasks.

The research problem is that hardness, usefulness, and validity are distinct.
A passage that resembles a query can be an informative negative, an
unlabeled relevant passage, or a lexical match with little useful contrastive
content. Conversely, a lexically novel passage can supply useful supervision
or simply be unrelated. False negatives can bias contrastive learning
\cite{bonifacio2022inpars,chuang2020debiased}, so a statistic that ranks
curated sources well need not establish that their negatives are valid.
We investigate this distinction rather than assume that a high diagnostic
score implies useful supervision. Our study asks three questions:\\
\textbf{RQ1:} How do lexical and frozen-encoder semantic statistics
associate with aggregate transfer across negative sources?\\
\textbf{RQ2:} Does combining semantic attributes with lexical
residuality, or adding residual geometry, improve observed rank agreement?\\
\textbf{RQ3:} How do these statistics behave when sources contain
easy, diluted, or false-negative-shaped supervision?\\

We use approximately $400{,}000$ MS MARCO query-positive pairs to construct
three individual sources and three pairwise hybrids. Each source trains
DistilBERT, E5-base, and Contriever under the same within-setting protocol.
We compare pre-training diagnostics with mean nDCG@10 on eleven BEIR
datasets. The comparison separates lexical noncoverage from semantic
consistency, passage locality, and hardness. We then combine these attributes
and vary their aggregation to ask whether the more elaborate descriptions
change the source ordering. The role of the composite is to test that
question, rather than to introduce a source-selection method.

The central observation is that the lexical statistic already recovers the
complete transfer ordering in every block. The semantic statistics alone
do not, and removing either semantic factor from their combination with
lexical residuality leaves the ordering unchanged. Constructed unreliable
sources nevertheless receive favorable lexical and semantic values. Thus,
association with transfer and evidence of supervision validity must be
examined separately.

Our contributions are an empirical comparison of pre-training signals
against transfer, a component analysis showing when a complex diagnostic
reduces to a simpler statistic, and a controlled examination of failures
outside curated source comparisons. The evidence covers two distinct source
orderings replicated across three backbones, not six independent successes.
The study does not establish a causal explanation, verify false negatives through
human judgments, or test whether semantic signals add information after
lexical matching.

\section{Related Work}
Negative sampling for dense retrieval includes random, static, dynamic,
and synthetic sources \cite{wischounig2026negativesampling}. Dual encoders
and contrastive representation learning rely on the structure of negative
supervision \cite{karpukhin-etal-2020-dense,gao2021simcse,xiong2021approximate}.
Condenser, Contriever, GTR, and E5 connect representation quality to
pre-training and contrastive data
\cite{gao2021condenserpretrainingarchitecturedense,izacard2022contriever,ni2022gtr,wang2024textembeddingsweaklysupervisedcontrastive}.
Our focus is the association between source statistics before training and
transfer after training.

BM25 and dense mining define different notions of hardness
\cite{robertson2009probabilistic,zhan2021optimizing}. ANCE, STAR, and ADORE
refresh difficult examples \cite{xiong2021approximate,zhan2021optimizing};
TAS-B, TriSampler, and SimANS seek informative negatives
\cite{hofstatter2021efficientlyteachingeffectivedense,yang2022trisampler,zhou2022simans}.
False-negative handling is a separate concern
\cite{robinson2021contrastive,chuang2020debiased,wang2023mitigating,moreira2024nvretriever}.
Our controls test selected weightings of frozen statistics. They do not
establish a general impossibility result for these methods or for filters
that use additional evidence.

Synthetic supervision includes InPars, GPL, Promptagator, and generated
hard negatives
\cite{bonifacio2022inpars,wang2022gplgenerativepseudolabeling,dai2023promptagator,li2024synegllmdrivensynthetichardnegatives,sinha2025dontretrievegenerateprompting}.
Structured evidence also supports retrieval training
\cite{sinha2025bicaeffectivebiomedicaldense}.
We compare one generator with one dense miner and BM25; broader conclusions
about generator or miner families require additional sources.

Geometric aggregation is one way to combine source attributes. Log-determinants appear in
determinantal processes, experimental design, information acquisition, and
representation covariance objectives
\cite{kulesza2012determinantal,chaloner1995bayesian,pukelsheim2006optimal,mackay1992information,foster2019variational,krause2014submodular,bardes2022vicreg}.
These connections distinguish the aggregation hypothesis from the lexical
and semantic attributes being aggregated. Our study asks whether changing
that aggregation changes the observed source ordering.

\section{Study Design}
\label{sec:study_design}
\subsection{Negative Sources}
All sources use approximately $400{,}000$ query-passage pairs from MS MARCO
\cite{bajaj2018msmarcohumangenerated}. BM25S \cite{lù2024bm25sordersmagnitudefaster} indexes the corpus; we retain the top ten passages per query after removing labeled positives. For dense mining, frozen \texttt{msmarco-MiniLM-L6-v3} embeddings are normalized and the top ten non-positive passages by cosine similarity are retained.

We use the \texttt{Qwen3-30B-A3B-Thinking-2507}
\cite{yang2025qwen3technicalreport}, with vLLM
\cite{kwon2023efficient} on two NVIDIA L40s GPUs. Given the query and
positive, its prompt requests ten JSON-formatted passages that share topical
cues and initially appear relevant, but omit or contradict the information
need while remaining plausible and diverse. This is a generation instruction,
not a verified property of every output.
The generation prompt is reproduced in Appendix~\ref{sec:llm_hard_negative_prompt}.

Individual sources contain $K=10$ negatives per record, approximately
$4$M flattened triplets. Pairwise hybrids combine BM25+Dense, BM25+LLM,
and Dense+LLM with $K=20$, approximately $8$M triplets. We compare rankings
within each setting because source composition and training volume both
change between individual and hybrid runs.

\subsection{Training and Transfer Evaluation}
\label{sec:training}
Each source fine-tunes DistilBERT
\cite{sanh2020distilbertdistilledversionbert}, E5-base
\cite{wang2024textembeddingsweaklysupervisedcontrastive}, and Contriever
\cite{izacard2022contriever} for one epoch. We use maximum sequence length
$128$, learning rate $2\times10^{-5}$, warmup ratio $0.1$, and mixed precision
via bfloat16 or fp16 when available. CachedMultipleNegativesRankingLoss
\cite{reimers-2019-sentence-bert} uses GradCache
\cite{gao2021scaling}, \texttt{query\_to\_doc},
\texttt{BatchSamplers.NO\_DUPLICATES}, batch size $B=4096$, and encoder
mini-batches of $2048$.

For flattened triplets $(q_i,p_i,n_i)$, training pools all batch positives
and explicit negatives in the denominator:
\begin{equation}
\begin{aligned}
\mathcal L&=-\frac1B\sum_{i=1}^B
\log\frac{\exp(s(q_i,p_i)/\tau)}
{\sum_{d\in\mathcal D_B}\exp(s(q_i,d)/\tau)},\\
\mathcal D_B&=\{p_j\}_{j=1}^B\cup\{n_j\}_{j=1}^B,\quad \tau=0.05.
\end{aligned}
\label{eq:fine_tuning_mnrl}
\end{equation}
Transfer is measured by nDCG@10 on eleven BEIR datasets
\cite{thakur2021beir}, without training on those datasets. The aggregate
outcome is the unweighted mean of the listed dataset scores. Runs use a
single training seed; we report observed ordering rather than statistical
significance or a fitted effect-size relationship.

\subsection{Scoring Protocol and Comparison Units}
\label{sec:training_free_ablations}
Scoring requires a query, a labeled positive, and explicit candidate
negatives; it is training-free, not label-free. We use the first labeled
positive when several are available. Frozen encoder scoring matches training
formatting, sequence length, normalization, and temperature, including E5
\texttt{query:} and \texttt{passage:} prefixes. Full-file composite scores
are supplemented by ablations on $25{,}000$ uniformly sampled valid records
per source with seed $42$.

Baseline diagnostic comparisons, gate ablations, and controls use the frozen
backbone that initializes the corresponding training run. Other reported
robustness checks use frozen \texttt{prdev/mini-gte}
embeddings\footnote{\url{https://huggingface.co/prdev/mini-gte}}, sequence
length $128$, first-positive aggregation, and $\tau=0.05$ unless varied.
Each backbone and setting defines a block of three sources. Agreement means
recovering the complete mean-transfer ordering within a block. Correlations
are computed over only three sources per block and are descriptive.

\subsection{Model Training}
We fine-tune three retriever backbones: DistilBERT (66M parameters),
E5-base-v2 (approximately 110M), and Contriever (approximately 110M).
All fine-tuning runs use a single NVIDIA RTX PRO 6000 GPU with 96\,GB
memory. Models are trained for one epoch with maximum sequence length
128, learning rate $2\times10^{-5}$, and warmup ratio 0.1. We use
\texttt{CachedMultipleNegativesRankingLoss} with GradCache in the
query-to-document direction, an effective batch size of 4096, and
encoder mini-batches of 2048. Batches use
\texttt{BatchSamplers.NO\_DUPLICATES} so that duplicate examples are
not introduced as in-batch negatives. Training uses bfloat16 for all models and initializes everything through Sentence-Transformers\cite{reimers-2019-sentence-bert}.

\section{Lexical and Semantic Source Attributes}
\label{sec:method_eci}
Let $u$, $v^+$, and $v^-$ be unit-normalized frozen embeddings for a query,
positive, and negative. We compare lexical noncoverage, semantic similarity, and hardness, then
test combinations of these attributes. The statistics describe source
composition; none is interpreted as a relevance label.

\subsection{Lexical and Semantic Statistics}
Let $T(x)$ be normalized alphanumeric tokens, with nonnegative smoothed IDF
weights computed on a fixed candidate scoring corpus. Query-token coverage
and lexical residuality are
\begin{equation}
\begin{aligned}
C(q,n)&=\frac{\sum_{t\in T(q)\cap T(n)}\operatorname{idf}(t)}
{\sum_{t\in T(q)}\operatorname{idf}(t)},\\
\psi&=1-C(q,n).
\end{aligned}
\label{eq:eci_lexical_residual}
\end{equation}
For a zero denominator, $C=0$. Higher $\psi$ means lower IDF-weighted
query-token coverage; it does not imply relevance or difficulty.

Two frozen-encoder gates are
\begin{align}
\rho&=\sigma\big((u^\top v^+-u^\top v^-)/\tau\big),
\label{eq:eci_pairwise_win}\\
\eta&=\sigma\big(((v^+)^\top v^- -u^\top v^-)/\tau\big).
\label{eq:eci_semantic_neighborhood}
\end{align}
Here $\sigma$ is the logistic sigmoid. $\rho$ is the probability of the
labeled positive in a two-document softmax. It increases as the encoder
prefers the positive, including when the negative becomes trivially easy.
$\eta$ measures negative-positive proximity relative to negative-query
similarity. Neither gate verifies relevance. We report their means, mean
$\psi$, and pairwise inversion rate $\Pr[\rho<0.5]$.

The baseline comparison also includes reported local gradient energy, an
MNRL loss proxy, margin-hardness, and record ambiguity. Ambiguity is the
fraction of records with at least one explicit negative ranked above the
positive. It differs from the pairwise inversion rate and depends on $K$.
Table~\ref{tab:ablation_baseline_diagnostics} reports the baseline values. These source diagnostics differ from the full training objective in
Eq.~\eqref{eq:fine_tuning_mnrl}. Proposition~\ref{prop:pairwise_gradient_information}
derives the local two-document gradient, separating its hardness weighting
from the composite diagnostic.

Appendices~\ref{sec:target_contrastive_setting}, \ref{sec:pairwise_target_consistency}, and \ref{sec:semantic_residual_information} give the target setting, gate, and residual definitions.

\subsection{Combining Attributes and Testing Geometry}
To compare combined attributes with the individual statistics, we use the
weight $w=\rho\eta\psi$. A weighted trace averages $w\|r\|_2^2$ over
explicit negatives, where $r$ is the unit positive-negative embedding
residual, set to zero for identical embeddings. A log-determinant variant
uses $\widehat{\mathcal I}=N^{-1}\sum wrr^\top$ to summarize their
directional distribution. Appendix~\ref{sec:theory} gives the complete
definitions and structural properties of this residual composite, denoted
\eci{} there. We use its components and trace comparison to distinguish
attribute weighting from geometric aggregation. Neither construction is
treated as a separate empirical contribution.

\section{Lexical versus Semantic Associations with Transfer}
\label{sec:results}
\subsection{Source Orderings (RQ1)}
Table~\ref{tab:transfer_summary} places lexical residuality beside aggregate
transfer; Table~\ref{tab:ndcg10-beir} reports the complete BEIR results. Within individual sources, the mean-transfer ordering is
LLM $>$ Dense $>$ BM25 for all three backbones. Within hybrids, it is
Dense+LLM $>$ BM25+LLM $>$ BM25+Dense. Mean lexical residuality agrees
with all eighteen within-block pairwise comparisons. These are two source orderings replicated
across backbones, with shared data and sources, not independent trials.

\begin{table*}[t]
    \centering
    \scriptsize
    \setlength{\tabcolsep}{1.2pt}
    \renewcommand{\arraystretch}{0.85}

    \resizebox{\textwidth}{!}{%
    \begin{tabular}{lllcccccccccccc}
      \toprule
      \multirow{2}{*}{\textbf{Model}} &
      \multirow{2}{*}{\textbf{Setting}} &
      \multirow{2}{*}{\textbf{Mining}} &
      \multicolumn{11}{c}{\textbf{Dataset}} &
      \multirow{2}{*}{\textbf{Mean}} \\
      \cmidrule(lr){4-14}
      & & & ArguAna & NFCorpus & Quora & SciDocs & SciFact &
      \shortstack{TREC\\COVID} &
      \shortstack{Webis\\Touche} & FiQA & NQ & FEVER &
      \shortstack{Climate\\FEVER} & \\
      \midrule
      \multirow{6}{*}{DistilBERT}
        & \multirow{3}{*}{Non-hybrid}
        & BM25  & 0.328 & 0.250 & 0.827 & 0.132 & 0.414 & 0.490 & 0.151 & 0.104 & 0.342 & 0.606 & 0.141 & 0.344 \\
        & & Dense & \textbf{0.464} & \textbf{0.266} & \textbf{0.840} & 0.125 & \textbf{0.492} & 0.452 & 0.159 & 0.211 & 0.315 & 0.620 & 0.164 & 0.373 \\
        & & LLM   & 0.456 & 0.262 & 0.836 & \textbf{0.135} & 0.488 & \textbf{0.493} & \textbf{0.190} & \textbf{0.221} & \textbf{0.368} & \textbf{0.668} & \textbf{0.197} & \textbf{0.392} \\
      \cmidrule(lr){2-15}
        & \multirow{3}{*}{Hybrid}
        & BM25+Dense & 0.369 & 0.258 & 0.833 & 0.122 & 0.424 & 0.500 & 0.141 & 0.108 & 0.330 & 0.601 & 0.150 & 0.349 \\
        & & BM25+LLM   & 0.357 & 0.258 & 0.829 & \textbf{0.130} & 0.435 & \textbf{0.520} & 0.155 & 0.124 & \textbf{0.366} & 0.626 & 0.145 & 0.359 \\
        & & Dense+LLM  & \textbf{0.452} & \textbf{0.269} & \textbf{0.839} & 0.126 & \textbf{0.493} & 0.479 & \textbf{0.172} & \textbf{0.211} & 0.343 & \textbf{0.631} & \textbf{0.179} & \textbf{0.381} \\
      \midrule
      \multirow{6}{*}{E5-base}
        & \multirow{3}{*}{Non-hybrid}
        & BM25  & \textbf{0.538} & 0.322 & 0.750 & 0.172 & 0.600 & 0.728 & \textbf{0.265} & 0.158 & 0.456 & \textbf{0.757} & 0.225 & 0.452 \\
        & & Dense & 0.521 & 0.332 & \textbf{0.848} & 0.165 & 0.616 & 0.716 & 0.243 & \textbf{0.326} & 0.456 & 0.728 & 0.167 & 0.465 \\
        & & LLM   & 0.502 & \textbf{0.336} & 0.846 & \textbf{0.177} & \textbf{0.618} & \textbf{0.754} & 0.255 & 0.319 & \textbf{0.483} & \textbf{0.757} & \textbf{0.227} & \textbf{0.479} \\
      \cmidrule(lr){2-15}
        & \multirow{3}{*}{Hybrid}
        & BM25+Dense & \textbf{0.525} & 0.328 & 0.729 & 0.161 & 0.590 & 0.711 & 0.264 & 0.183 & 0.446 & 0.714 & 0.176 & 0.439 \\
        & & BM25+LLM   & 0.518 & 0.328 & 0.739 & \textbf{0.168} & 0.597 & \textbf{0.727} & \textbf{0.280} & 0.180 & \textbf{0.464} & \textbf{0.736} & \textbf{0.211} & 0.450 \\
        & & Dense+LLM  & 0.487 & \textbf{0.331} & \textbf{0.822} & 0.162 & \textbf{0.603} & 0.720 & 0.251 & \textbf{0.315} & 0.457 & 0.716 & 0.184 & \textbf{0.459} \\
      \midrule
      \multirow{6}{*}{Contriever}
        & \multirow{3}{*}{Non-hybrid}
        & BM25  & 0.376 & 0.292 & 0.840 & 0.155 & 0.442 & 0.467 & 0.147 & 0.149 & 0.390 & 0.453 & 0.125 & 0.349 \\
        & & Dense & \textbf{0.524} & \textbf{0.320} & \textbf{0.859} & 0.140 & 0.589 & 0.439 & 0.117 & 0.253 & 0.345 & 0.631 & 0.152 & 0.397 \\
        & & LLM   & 0.506 & \textbf{0.320} & 0.858 & \textbf{0.159} & \textbf{0.598} & \textbf{0.476} & \textbf{0.189} & \textbf{0.287} & \textbf{0.433} & \textbf{0.721} & \textbf{0.225} & \textbf{0.434} \\
      \cmidrule(lr){2-15}
        & \multirow{3}{*}{Hybrid}
        & BM25+Dense & 0.398 & 0.300 & 0.849 & 0.132 & 0.404 & 0.456 & 0.114 & 0.144 & 0.369 & 0.416 & 0.092 & 0.334 \\
        & & BM25+LLM   & 0.398 & 0.294 & 0.852 & \textbf{0.151} & 0.405 & \textbf{0.492} & \textbf{0.158} & 0.162 & \textbf{0.431} & 0.492 & 0.119 & 0.359 \\
        & & Dense+LLM  & \textbf{0.496} & \textbf{0.320} & \textbf{0.860} & 0.138 & \textbf{0.569} & 0.456 & 0.143 & \textbf{0.258} & 0.371 & \textbf{0.647} & \textbf{0.164} & \textbf{0.402} \\
      \bottomrule
    \end{tabular}
    }
\caption{BEIR nDCG@10 after MS MARCO fine-tuning. Results are grouped by
    backbone and by non-hybrid versus hybrid mining. Mean averages the listed
    datasets. Bold marks the best score within each block and column.}
\label{tab:ndcg10-beir}
\end{table*}

\begin{table}[t]
\centering\small\setlength{\tabcolsep}{3.5pt}

\begin{tabular}{lcccc}
\toprule Source & Mean $\psi$ & D & E & C\\\midrule
BM25 & .288 & .344 & .452 & .349\\
Dense & .387 & .373 & .465 & .397\\
LLM & .531 & .392 & .479 & .434\\\midrule
BM25+Dense & .338 & .349 & .439 & .334\\
BM25+LLM & .410 & .359 & .450 & .359\\
Dense+LLM & .461 & .381 & .459 & .402\\
\bottomrule
\end{tabular}
\caption{Lexical residuality and mean BEIR nDCG@10. Mean $\psi$ comes
from the $5{,}000$-row sensitivity sample with alphanumeric tokens and
pooled-document IDF (Table~\ref{tab:tokenizer_idf_sensitivity}). D/E/C:
DistilBERT, E5-base, and Contriever. Compare within each setting.}
\label{tab:transfer_summary}
\end{table}

Mean $\psi$ alone recovers all six block orderings
(Table~\ref{tab:gate_ablations}). Mean $\rho$, mean $\eta$, and inversion
rate recover none. Ranking sources by larger MNRL proxy or record ambiguity
selects Dense in every individual-source block and BM25+Dense in every
hybrid block, rather than the highest-transfer source
(Table~\ref{tab:ablation_baseline_diagnostics}). Gradient energy is less consistent.
Thus, greater measured hardness does not identify the best-transfer source
in these comparisons. The lexical statistic supplies a simpler baseline
with the same observed rank agreement as the composite.

This is aggregate agreement, not per-dataset dominance. With E5-base,
BM25 exceeds LLM on ArguAna ($0.538$ versus $0.502$) and Touch\'e
($0.265$ versus $0.255$), while Dense exceeds LLM on Quora and FiQA.
Dense also exceeds LLM on ArguAna and Quora with DistilBERT and Contriever.
Table~\ref{tab:ndcg10-beir} reports all eleven dataset scores for every training condition.
These exceptions do not by themselves identify lexical overlap as their
cause. Source choice for a particular task still requires task-specific
evaluation.

\paragraph{Magnitude and task dependence.}
The aggregate scores make the source differences concrete. In the individual
setting, LLM exceeds BM25 by $0.048$, $0.027$, and $0.085$ mean nDCG@10
for DistilBERT, E5-base, and Contriever, respectively; its advantage over
Dense is $0.019$, $0.014$, and $0.037$. These differences are computed from
the reported rounded means and are descriptive, without repeated-seed
uncertainty estimates. The contrast between a common ordering and different
gap sizes also motivates evaluating ranking agreement separately from
associating a statistic with the magnitude of a training gain.

Among hybrids, Dense+LLM exceeds BM25+Dense by $0.032$, $0.020$, and
$0.068$ on the three backbones. Individual LLM also has a higher aggregate
score than any hybrid for each backbone, despite using half the flattened
triplets. This observation shows that more supervision alone does not
explain the strongest result in these runs. It does not isolate the effect
of mixing sources: a budget-matched comparison would be needed for that
claim. The per-dataset reversals in Table~\ref{tab:ndcg10-beir} further
show that the aggregate winner is not uniformly best across retrieval tasks.

\begin{table*}[t]
    \centering
    \small
    \setlength{\tabcolsep}{3pt}
    \renewcommand{\arraystretch}{0.9}

    \begin{tabular}{lllcccc}
        \toprule
        \textbf{Model} & \textbf{Setting} & \textbf{Source} &
        \textbf{Composite} &
        \textbf{Mean $\|g_{j,k}\|_2^2$} &
        \textbf{MNRL proxy} & \textbf{Ambiguity (any-of-$K$)} \\
        \midrule
        \multirow{6}{*}{DistilBERT}
            & \multirow{3}{*}{Non-hybrid}
            & BM25  & 0.1282 & 27.59 & 2.416 & 0.881 \\
            & & Dense & 0.1760 & 22.96 & \textbf{2.434} & \textbf{0.883} \\
            & & LLM   & \textbf{0.2610} & \textbf{29.83} & 2.184 & 0.733 \\
        \cmidrule(lr){2-7}
            & \multirow{3}{*}{Hybrid}
            & BM25+Dense & 0.1521 & 25.41 & \textbf{3.077} & \textbf{0.934} \\
            & & BM25+LLM   & 0.1947 & \textbf{28.66} & 2.955 & 0.911 \\
            & & Dense+LLM  & \textbf{0.2193} & 25.55 & 2.958 & 0.915 \\
        \midrule
        \multirow{6}{*}{E5-base}
            & \multirow{3}{*}{Non-hybrid}
            & BM25  & 0.0965 & \textbf{39.86} & 1.790 & 0.533 \\
            & & Dense & 0.1193 & 39.61 & \textbf{2.107} & \textbf{0.689} \\
            & & LLM   & \textbf{0.2108} & 29.56 & 1.031 & 0.122 \\
        \cmidrule(lr){2-7}
            & \multirow{3}{*}{Hybrid}
            & BM25+Dense & 0.1079 & \textbf{42.96} & \textbf{2.584} & \textbf{0.715} \\
            & & BM25+LLM   & 0.1537 & 42.68 & 2.060 & 0.547 \\
            & & Dense+LLM  & \textbf{0.1657} & 40.86 & 2.297 & 0.690 \\
        \midrule
        \multirow{6}{*}{Contriever}
            & \multirow{3}{*}{Non-hybrid}
            & BM25  & 0.0671 & \textbf{162.68} & 2.573 & 0.801 \\
            & & Dense & 0.1170 & 155.55 & \textbf{2.719} & \textbf{0.840} \\
            & & LLM   & \textbf{0.1807} & 91.90 & 1.033 & 0.344 \\
        \cmidrule(lr){2-7}
            & \multirow{3}{*}{Hybrid}
            & BM25+Dense & 0.0921 & \textbf{166.39} & \textbf{3.326} & \textbf{0.882} \\
            & & BM25+LLM   & 0.1240 & 163.00 & 2.710 & 0.811 \\
            & & Dense+LLM  & \textbf{0.1494} & 155.74 & 2.829 & 0.845 \\
        \bottomrule
    \end{tabular}
\caption{Reported frozen target-backbone diagnostics on $25{,}000$
seeded records per source. Ambiguity is the record-level any-of-$K$
inversion rate. Gradient energy and MNRL proxy retain their original
reported values. The pairwise gradient derivation appears in Appendix~\ref{sec:loss_reduction_information_volume}. Bold marks column maxima within backbone and setting.}
\label{tab:ablation_baseline_diagnostics}
\end{table*}

\subsection{Do Semantic Attributes Change the Lexical Ordering? (RQ2)}
\label{sec:linear_regime}
Table~\ref{tab:gate_ablations} separates individual lexical and semantic
statistics from combinations. Mean $\psi$ recovers all six complete block
orderings, whereas mean $\rho$, mean $\eta$, and inversion rate recover
none. Dropping $\rho$ or $\eta$ from the combined weighting leaves all
six orderings unchanged; dropping $\psi$ changes two. The observed
agreement therefore depends on lexical residuality, while either semantic
factor can be removed without changing it.

Mean $\psi$ has reported mean within-block Pearson correlations of $0.96$
in both settings, compared with $0.93$ and $0.95$ for the weighted composite.
Each correlation uses only three sources per backbone. The evidence is the
ordering and component comparisons, not the small differences between these
correlations. The checks do not estimate the information carried by semantic
statistics conditional on lexical residuality; that requires lexical-matched
sources or a separate controlled analysis.

\paragraph{Attribute weighting versus geometry.}
At the default scale, weighted log-determinant and trace differ by at most
$0.21\%$ relative (median $0.03\%$) and induce identical rankings. The gap
is below $0.9\%$ for bounded gate ablations, versus roughly $15\%$ for the
larger margin-hardness weighting. Thus, the default composite supplies
almost the same numerical summary as the mean weighted residual norm.

A 48-point scale sweep from $c=1$ to $325.5$ in
$\log\det(I_d+c\widehat{\mathcal I})$ also preserves the full-weighting
rankings in all $288$ block-scale comparisons, including scales with
$c\lambda_{\max}$ between $1.16$ and $6.51$. Uniform residual weighting
changes the ranking in $152$ comparisons and flips it in four blocks.
Geometry can change this statistic, but does not yield a different ordering
when applied with the combined attributes on the tested scale grid. These
results delimit its empirical role rather than establish a geometric
mechanism for transfer.

\paragraph{Lexical-residuality-matched analysis.}
To test whether the frozen-encoder statistics provide additional ordinal information after controlling for lexical residuality, we matched BM25, Dense, and LLM negatives by $\psi$. We restricted the samples to common support, divided $\psi$ into equal-width bins, and sampled the same number of negatives from each source within every bin. With 20 bins, each source retained 115,869 negatives (46.35\%), with mean $\psi$ of 0.3949, 0.3956, and 0.3969 for BM25, Dense, and LLM, respectively. We repeated the matching 100 times and evaluated mean $\rho$, mean $\eta$, and inversion rate using each frozen backbone. None of the three statistics recovered the observed LLM $>$ Dense $>$ BM25 transfer ordering in any repetition. The result is unchanged with 10 or 50 bins. Thus, approximately equalizing lexical residuality does not reveal additional ordinal agreement from the tested frozen-encoder statistics. This is a diagnostic comparison on fixed trained models and does not establish a causal effect of lexical residuality on training.

\begin{table}[!ht]
\centering
\small
\setlength{\tabcolsep}{2.5pt}
\renewcommand{\arraystretch}{0.92}
\resizebox{\columnwidth}{!}{%
\begin{tabular}{lcccc}
\toprule
Variant & Criterion (D/E/C) & Blocks & $r$ (NH) & $r$ (H) \\
\midrule
Mean $\psi$ only & n/a & 6/6 & 0.96 & 0.96 \\
Mean $\rho$ only & n/a & 0/6 & 0.62 & 0.69 \\
Mean $\eta$ only & n/a & 0/6 & 0.12 & 0.03 \\
Inversion rate & n/a & 0/6 & $-0.58$ & $-0.67$ \\
\midrule
Combined attributes (log-det) & no/no/no & 6/6 & 0.93 & 0.95 \\
Trace form & no/no/no & 6/6 & 0.93 & 0.95 \\
Drop $\rho$ ($\eta\psi$) & no/no/no & 6/6 & 0.98 & 0.97 \\
Drop $\eta$ ($\rho\psi$) & no/no/no & 6/6 & 0.89 & 0.92 \\
Drop $\psi$ ($\rho\eta$) & no/no/no & 4/6 & 0.78 & 0.85 \\
Uniform weights (residual only) & no/\textbf{yes}/no & 2/6 & 0.61 & 0.72 \\
Gradient $(1-\rho)^2\eta\psi$ & no/no/no & 2/6 & 0.14 & 0.04 \\
Margin-hardness & n/a & 0/6 & $-0.55$ & $-0.71$ \\
$\rho^2(1-\rho)$ & no/no/no & 2/6 & $-0.23$ & $-0.29$ \\
$\rho^2(1-\rho)\eta$ & no/no/no & 0/6 & $-0.08$ & $-0.14$ \\
$\rho(1-\rho)\psi$ & no/no/no & 2/6 & 0.35 & 0.24 \\
$\rho(1-\rho)\eta\psi$ & no/no/no & n/a & n/a & n/a \\
$\rho\eta\psi\,\mathbf{1}[\rho<0.99]$ & no/no/no & n/a & n/a & n/a \\
$\rho\eta\psi\,\mathbf{1}[\rho<0.999]$ & no/no/no & n/a & n/a & n/a \\
\bottomrule
\end{tabular}%
}
\caption{Component and weighting comparisons on $25{,}000$ records.
Blocks counts complete transfer orderings recovered, out of six.
$r$ is Pearson correlation with mean BEIR nDCG@10 over three sources per
backbone and setting, averaged across backbones (NH: individual, H: hybrid).
Control criterion requires all eight controls to rank below all six genuine
sources; D/E/C denote the three backbones. Only one of the 36 tested
backbone-weighting cells meets it. n/a denotes an untested comparison.
Trace includes the zero-residual convention.}
\label{tab:gate_ablations}
\end{table}

\subsection{Robustness of the Observed Rankings}
\label{sec:ablation_sample_size_stability}
The component comparison raises a practical question: are the source
orderings sensitive to sampling or to the implementation of lexical overlap?
The robustness experiments vary these choices without retraining
the downstream models. At $25\%$, $50\%$, $75\%$, and $100\%$ of each
seeded sample, the composite recovers exactly the full-sample ranking
(Table~\ref{tab:sample_size_rank_stability}). Thus, the observed ranking
does not require the largest scoring sample used in this experiment.

Sensitivity checks use $5{,}000$ seeded rows per source with frozen
\texttt{prdev/mini-gte}\footnote{\url{https://huggingface.co/prdev/mini-gte}}\cite{li2023generaltextembeddingsmultistage}. Four tokenizer settings compare alphanumeric tokens, tokens of length at least two, whitespace tokens, and character trigrams. Four lexical-weight settings compare pooled-document IDF,
source-specific IDF, pooled query-document IDF, and uniform weights.
All preserve the within-setting source orderings
(Table~\ref{tab:tokenizer_idf_sensitivity}). Scores nevertheless change:
BM25 ranges from $0.0884$ to $0.1020$ across IDF choices, whereas LLM
ranges from $0.1857$ to $0.1866$. These differences show why absolute
score calibration and rank stability should be examined separately.

Failure buckets also help interpret the gates. In the same sensitivity
sample, high query-token coverage ($C\ge0.50$) occurs in $89.6\%$ of
BM25 negatives, $69.3\%$ of Dense negatives, and $38.0\%$ of LLM
negatives. Pairwise inversions occur in $16.8\%$, $33.5\%$, and
$2.1\%$, respectively (Table~\ref{tab:eta_psi_failure_buckets}). These
are mini-gte measurements, distinct from the target-backbone control
profiles. They describe source composition rather than certify the
relevance of individual passages. Finally, the reported temperature sweep
over $\tau\in\{0.03,0.05,0.1\}$ preserves the source orderings.
Together, the checks support the stability of the present comparisons under
several scoring choices; they do not enlarge the set of training corpora or
independent negative sources studied.
Appendix~\ref{sec:appendix_robustness} details sample size (\ref{sec:appendix_sample_size_stability}), tokenizer/IDF (\ref{sec:ablation_tokenizer_idf_failure}), and temperature (\ref{sec:ablation_temperature_sweep}).

\section{Do Transfer-Associated Signals Indicate Valid Supervision? (RQ3)}
\label{sec:controls}
\subsection{Constructed Controls}
Eight controls are scored alongside six genuine sources on $25{,}000$
seeded records per source with each target backbone: unrelated random
passages; negatives shuffled across unrelated records; a source replacing
half its negatives with the record's positive; a source replacing each
negative with another query's labeled positive selected by proximity to
the record's positive; and four BM25 sources diluted with $25$, $50$, $75$,
or $100$ percent random passages.

The neighbor-positive control is \emph{false-negative-shaped}. Its passages
are positives for other queries and close to the target positive, but have
not been relevance-judged for the target query. It tests susceptibility to
plausibly relevant contamination, not performance on verified false negatives.
We define a stringent control-ranking criterion: every control must rank
below every genuine source. This criterion also penalizes easy or diluted
supervision; it is not a universal definition of safety.

Under the combined attribute weighting, six controls outrank the best genuine source
on every backbone. LLM ranks seventh of fourteen, and the neighbor-positive
control ranks first on Contriever (Table~\ref{tab:control_gates}).
Own-positive corruption ranks last; exact positive duplicates contribute
zero residual rather than being detected as invalid labels by a gate.

The component values explain why high scores are ambiguous.
Random passages have mean $\psi=0.943$, versus $0.287$ for BM25.
Mean $\rho$ for random passages is approximately $0.77$, $0.98$, and $0.99$
on DistilBERT, E5-base, and Contriever, versus $0.53$, $0.68$, and $0.60$
for BM25. The neighbor-positive control also receives large locality
values. Lower query overlap, positive preference, and proximity in passage
space can occur without useful negative supervision. The lexical association
observed within genuine sources therefore fails as an unrestricted reward.

\begin{table}

\centering
\small
\setlength{\tabcolsep}{2.5pt}
\begin{tabular}{l c c c}
\toprule
Source & DistilBERT & E5-base & Contriever \\
\midrule
Easy shuffled & 0.602 (1) & 0.444 (3) & 0.443 (4) \\
Random passages & 0.602 (3) & 0.444 (1) & 0.444 (3) \\
Neighbor positive & 0.395 (5) & 0.360 (5) & \textbf{0.500 (1)} \\
LLM & 0.261 (7) & 0.211 (7) & 0.181 (7) \\
BM25 & 0.128 (13) & 0.097 (13) & 0.067 (13) \\
Own-positive mix & 0.064 (14) & 0.048 (14) & 0.034 (14) \\
\bottomrule
\end{tabular}
\caption{Selected sources ranked by the combined attribute score at
$25{,}000$ records, with rank out of fourteen in parentheses. Six controls
outrank the strongest genuine source on every backbone. Omitted ranks belong to
the random-mixing ladder; rank 2 on every backbone is its $100$ percent
endpoint, an independent draw from the same distribution as the
random-passage control. Values are printed to three decimals; near-ties among
the degenerate controls resolve only at higher precision (ranks 1 to 3 on
DistilBERT span $0.6017$ to $0.6022$).}
\label{tab:control_gates}
\end{table}

\subsection{Alternative Weightings}
Table~\ref{tab:gate_ablations} includes the twelve tested weightings and
their control-ranking verdicts. Only one of thirty-six backbone-weighting
cells satisfies the criterion: uniform residual weighting on E5-base.
That variant recovers two of six transfer block orderings overall.
The neighbor-positive control contributes to thirty-three of thirty-five
failures and is the sole violating control in thirteen. Two-sided
$\rho$ weightings reduce some easy-source scores but still fail the
control-ranking criterion. Three schemes passing a $5{,}000$-record
Contriever pilot fail at $25{,}000$ under the six-source comparison.

This is a failure of the tested statistics and weightings under a specified
stress test. It does not prove that no function of frozen embeddings can
identify false negatives, nor that other margin filters inherit this
failure. External judgments, cross-encoders, or additional features could
supply evidence absent from these comparisons.

\subsection{Source Profiles Under Constructed Contamination}
The same statistics can describe departures from genuine sources without
being maximized as a reward. A source profile contains mean $\rho$, mean
$\eta$, mean $\psi$, and inversion rate. Thresholds are the extrema of
the six genuine sources for each backbone: flag a mean exceeding the
largest genuine-source mean or an inversion rate below the smallest.

Table~\ref{tab:audit} gives the Contriever profile. Seven of eight controls
are flagged, with the same flags and miss on DistilBERT and E5-base.
The $25\%$ random mixture stays inside the reference envelope and is missed.
Both corrupted-label controls are flagged by $\eta$. No genuine source is
flagged \emph{by construction}, because those sources define the thresholds;
this is not an estimated false-positive rate on held-out sources. The audit
demonstrates descriptive separation under constructed contamination. It is
a reason to seek independent relevance evidence, not a correction that
turns the ranking score into a validity detector.

\begin{table}

\centering
\footnotesize
\setlength{\tabcolsep}{3pt}
\renewcommand{\arraystretch}{0.92}
\begin{tabular}{lccccl}
\toprule
Source & $\overline{\rho}$ & $\overline{\eta}$ & $\overline{\psi}$ & inv. & Flags \\
\midrule
BM25 & 0.602 & 0.352 & 0.287 & 0.378 & none \\
Dense & 0.566 & 0.443 & 0.386 & 0.423 & none \\
LLM & 0.870 & 0.375 & 0.529 & 0.083 & none \\
BM25+Dense & 0.584 & 0.398 & 0.337 & 0.400 & none \\
BM25+LLM & 0.736 & 0.364 & 0.408 & 0.230 & none \\
Dense+LLM & 0.720 & 0.409 & 0.459 & 0.251 & none \\
\midrule
Random passages & 0.994 & 0.471 & 0.943 & 0.000 & $\rho$ $\eta$ $\psi$ inv \\
Easy shuffled & 0.994 & 0.469 & 0.944 & 0.000 & $\rho$ $\eta$ $\psi$ inv \\
Mix random $100\%$ & 0.993 & 0.471 & 0.943 & 0.001 & $\rho$ $\eta$ $\psi$ inv \\
Mix random $75\%$ & 0.915 & 0.447 & 0.812 & 0.076 & $\rho$ $\eta$ $\psi$ inv \\
Mix random $50\%$ & 0.798 & 0.412 & 0.615 & 0.188 & $\psi$ \\
Mix random $25\%$ & 0.680 & 0.376 & 0.418 & 0.302 & none \\
Neighbor positive & 0.831 & 0.753 & 0.704 & 0.116 & $\eta$ $\psi$ \\
Own-positive mix & 0.551 & 0.675 & 0.273 & 0.352 & $\eta$ \\
\bottomrule
\end{tabular}
\caption{The audit profile on all fourteen sources, in-family Contriever at
$25{,}000$ records. Thresholds come from the six real sources only: flag if
$\overline{\rho}>0.870$, $\overline{\eta}>0.443$, $\overline{\psi}>0.529$, or
inversion $<0.083$. Seven of eight controls are flagged and no genuine source is flagged by construction.}
\label{tab:audit}
\end{table}

\section{Discussion}
\label{sec:discussion}
\paragraph{The lexical association.}
Lexical residuality describes query-token noncoverage and requires no encoder. Across the genuine sources, mean $\psi$ reproduces both observed source orderings under all three backbones. We further approximately match the $\psi$ distributions of BM25, Dense, and LLM negatives and recompute the frozen-encoder statistics. Across 100 matching repetitions, three binning resolutions, and all three backbones, none of mean $\rho$, mean $\eta$, or inversion rate recovers the observed LLM $>$ Dense $>$ BM25 transfer ordering. This strengthens the evidence that lexical residuality provides the strongest ordinal agreement among the tested pre-training diagnostics. The association is not causal: mining and generation procedures also change passage content and relevance, and the constructed controls show that increasing lexical residuality alone does not guarantee useful supervision. We therefore treat $\psi$ as a simple source-comparison baseline rather than an objective to maximize.

\paragraph{The role of semantic information.}
The semantic attributes describe positive preference and relative passage
locality. Their failure to reproduce complete source orderings does not
exclude associations with particular transfer outcomes. Removing them
without changing rankings establishes no added ordinal benefit in the
tested comparisons; it does not establish statistical independence or
measure information conditional on lexical residuality. Lexical-matched
comparisons across independently varied source families would address that
question more directly.

The semantic statistics also characterize source reliability differently
from transfer. High locality and low inversion can distinguish unusual
source profiles even when rewarding those values yields undesirable
rankings. The audit uses these attributes descriptively, with relevance
judgments needed to determine whether flagged passages are valid negatives.
This role follows from the component measurements and controls, rather
than from the mathematical form of their weighted composite.


\section{Conclusion}
Lexical residuality tracks aggregate transfer across the studied MS MARCO
negative sources and three retriever backbones. The tested frozen-encoder
statistics do not recover the complete observed source orderings, and
combining them with lexical residuality does not improve rank agreement.
This remains true after approximately matching BM25, Dense, and LLM
diagnostic samples for lexical residuality: none of the semantic statistics
recovers the observed transfer ordering across the tested backbones and
matching resolutions. Constructed controls further show that favorable
lexical and semantic values can accompany unreliable supervision. These
results distinguish a source-level lexical association from both the
additional ordinal information provided by the tested semantic statistics
and the separate question of whether individual negatives provide reliable
supervision.
See Appendix~\ref{sec:full_file_composite_scores} and Table~\ref{tab:eci-lexdebias-full} for complete composite scores.

\section{Limitations}
The transfer comparisons use one fine-tuning run per source--backbone
condition, so training-run variance is not estimated. The constructed
controls are designed stress tests rather than human-verified false negatives
for the target queries. Accordingly, our conclusions concern the observed
source-level associations and the behavior of the tested diagnostics rather
than relevance judgments for individual negatives.

\bibliography{custom}
\appendix
\section{Full-File Composite Scores}
\label{sec:full_file_composite_scores}
Table~\ref{tab:eci-lexdebias-full} preserves full-file scores and the
reported per-dimension normalization. The main comparison uses
within-backbone rankings, not absolute cross-backbone magnitudes.
\begin{table}[!ht]
    \centering
    \small
    \setlength{\tabcolsep}{2pt}
    \renewcommand{\arraystretch}{0.8}

    \resizebox{\columnwidth}{!}{%
\begin{tabular}{lllcc}
      \toprule
      \textbf{Model} & \textbf{Setting} & \makecell{\textbf{Mining} \\ \textbf{method}} & $\mathbf{ECI}_{\mathrm{sem}} \uparrow$ & \makecell{$\mathbf{ECI}_{\mathrm{sem}}$ \\ / dim.$\uparrow$} \\
      \midrule
      \multirow{6}{*}{DistilBERT}
        & \multirow{3}{*}{Non-hybrid}
        & BM25  & 0.1265 & 0.000165 \\
        & & Dense & 0.1755 & 0.000228 \\
        & & LLM   & \textbf{0.2612} & \textbf{0.000340} \\
      \cmidrule(lr){2-5}
        & \multirow{3}{*}{Hybrid}
        & BM25+Dense & 0.1510 & 0.000197 \\
        & & BM25+LLM   & 0.1939 & 0.000252 \\
        & & Dense+LLM  & \textbf{0.2192} & \textbf{0.000285} \\
      \midrule
      \multirow{6}{*}{E5-base}
        & \multirow{3}{*}{Non-hybrid}
        & BM25  & 0.0954 & 0.000124 \\
        & & Dense & 0.1191 & 0.000155 \\
        & & LLM   & \textbf{0.2112} & \textbf{0.000275} \\
      \cmidrule(lr){2-5}
        & \multirow{3}{*}{Hybrid}
        & BM25+Dense & 0.1072 & 0.000140 \\
        & & BM25+LLM   & 0.1533 & 0.000200 \\
        & & Dense+LLM  & \textbf{0.1659} & \textbf{0.000216} \\
      \midrule
      \multirow{6}{*}{Contriever}
        & \multirow{3}{*}{Non-hybrid}
        & BM25  & 0.0667 & 0.000087 \\
        & & Dense & 0.1170 & 0.000152 \\
        & & LLM   & \textbf{0.1813} & \textbf{0.000236} \\
      \cmidrule(lr){2-5}
        & \multirow{3}{*}{Hybrid}
        & BM25+Dense & 0.0918 & 0.000120 \\
        & & BM25+LLM   & 0.1240 & 0.000161 \\
        & & Dense+LLM  & \textbf{0.1498} & \textbf{0.000195} \\
      \bottomrule
    \end{tabular}
    }
\caption{Training-free \eci{} scores, computed in family, that is, with a frozen
    copy of the same backbone that is later fine-tuned, over all explicit negatives in
    each source file. Bold marks the best score within each block.}
\label{tab:eci-lexdebias-full}
\end{table}

\section{Definitions and Properties of the Residual Composite}
\label{sec:theory}

This section specifies the residual composite, \eci{}, used in the
attribute and aggregation comparisons. Its mathematical properties explain
what the statistic summarizes. They do not establish an association with
transfer or constitute a separate optimization claim.

\subsection{Target Contrastive Setting}
\label{sec:target_contrastive_setting}

Each hard-negative file is flattened into triplets
\[
    (q_i,p_i,n_i),
\]
where $q_i$ is a query, $p_i$ is its labeled positive passage, and $n_i$ is one
explicit negative. A row with $K$ negatives contributes $K$ triplets, giving
per-negative comparisons across non-hybrid and hybrid sources.

Let $\phi_q(\cdot)$ and $\phi_d(\cdot)$ be the deterministic target formatting maps for
queries and passages, including role prefixes such as \texttt{query:} and
\texttt{passage:} when required by the encoder family. With frozen,
unit-normalized target-encoder embeddings,
\[
\begin{aligned}
u_i&=f_0(\phi_q(q_i)),\\
v_i^+&=f_0(\phi_d(p_i)),\\
v_i^-&=f_0(\phi_d(n_i)).
\end{aligned}
\]
where
\[
    \|u_i\|_2=\|v_i^+\|_2=\|v_i^-\|_2=1 .
\]

For a target MNRL batch
\[
    B=\{(q_i,p_i,n_i): i=1,\ldots,b\},
\]
define two document pools:
\[
\begin{aligned}
D_B^0&=\{p_1,\ldots,p_b\},\\
D_B^1&=\{p_1,\ldots,p_b,n_1,\ldots,n_b\}.
\end{aligned}
\]
$D_B^0$ is the positive-only baseline and $D_B^1$ adds explicit negatives.

For $s\in\{0,1\}$ and a document $d'\in D_B^s$, write
$v_{d'}=f_0(\phi_d(d'))$ for its target embedding and define
\[
    z_{i,d'}=\frac{u_i^\top v_{d'}}{\tau},
    \qquad
    \pi_i^s(d')
    =
    \frac{\exp(z_{i,d'})}{\sum_{c\in D_B^s}\exp(z_{i,c})}.
\]
The corresponding per-query target loss is
\begin{equation}
    \ell_i^s=-\log \pi_i^s(p_i),
    \qquad
    H_B^s=\sum_{i=1}^{b}\ell_i^s .
    \label{eq:baseline_and_augmented_losses}
\end{equation}
Fine-tuning uses the full in-batch softmax over $D_B^1$; \eci{} uses a scalable
pairwise restriction of the same target geometry.

For a fixed triplet, we write $\rho$, $\eta$, $\psi$, $w$, $r$, and
$\delta$ for its scalar factors, weight, residual direction, and embedding
difference, respectively.

\subsection{Pairwise Target-Consistency Gate}
\label{sec:pairwise_target_consistency}

Exact marginal changes from $H_B^0$ to $H_B^1$ require full in-batch
interactions. For scalable source screening, we use the frozen target encoder's
two-document judgment between the labeled positive and each explicit negative.

For a flattened triplet $(q,p,n)$, the two-document softmax probability assigned
to the labeled positive is
\begin{equation}
\begin{aligned}
    P_{\tau}(p \succ n \mid q)
    &=
    \frac{\exp(u^\top v^+/\tau)}
    {\exp(u^\top v^+/\tau)+\exp(u^\top v^-/\tau)}
    \\
    &=
    \sigma\!\left(
        \frac{u^\top v^+-u^\top v^-}{\tau}
    \right)
    \\
    &=
    \rho.
\end{aligned}
\label{eq:pairwise_surrogate_event}
\end{equation}
Thus $-\log\rho$ is exactly the two-document restriction of the MNRL loss
for the positive and the candidate negative.

$\rho$ is a target-consistency gate, not a gradient-magnitude proxy. It is large
when the frozen retriever assigns higher probability to the labeled positive
than to the candidate negative, reducing weight on inverted or potentially
mislabeled negatives.

The semantic-locality gate is
\begin{equation}
    \eta
    =
    \sigma\!\left(
        \frac{(v^+)^\top v^- - u^\top v^-}{\tau}
    \right).
    \label{eq:eci_semantic_neighborhood_theory}
\end{equation}
This heuristic favors negatives closer to the labeled positive passage than to
the query embedding under the frozen representation. It is a relative locality
gate, not an absolute nearest-neighbor claim.

To reduce sparse-retrieval shortcuts, we add a lexical-residual factor. Let
$T(x)$ denote the set of normalized alphanumeric tokens of a text $x$. We
compute nonnegative smoothed IDF weights over a fixed candidate scoring corpus,
for example
\[
    \operatorname{idf}(t)
    =
    \log\frac{M+1}{\operatorname{df}(t)+1}+1,
\]
where $M$ is the number of texts in the candidate scoring corpus and
$\operatorname{df}(t)$ is the number of such texts containing token $t$; under
this smoothing, $\operatorname{idf}(t)\ge 1$ for every token, so nonnegativity
holds automatically. The IDF-weighted lexical coverage of the query by the
negative and the associated lexical-residual factor are
\begin{equation}\begin{gathered}
    C(q,n)
    =
    \frac{
        \sum_{t\in T(q)\cap T(n)} \operatorname{idf}(t)
    }{
        \sum_{t\in T(q)} \operatorname{idf}(t)
    },
\\
    \psi=1-C(q,n).
    \label{eq:eci_lexical_residual_theory}
\end{gathered}\end{equation}
When the denominator is zero, we set $C(q,n)=0$ and $\psi=1$. Under this
convention, $C(q,n)\in[0,1]$ and $\psi\in[0,1]$.

Combining the three factors gives
\[
    w
    =
    \rho\,\eta\,\psi.
\]
$w$ is the target-consistent, semantically local, lexically residual weight of
an explicit negative.

\subsection{Semantic Residual Information}
\label{sec:semantic_residual_information}

The residual direction associated with a triplet is
\begin{equation}
    r
    =
    \frac{v^+ - v^-}{\|v^+ - v^-\|_2}.
    \label{eq:eci_pairwise_residual_theory}
\end{equation}
If $v^+=v^-$, set $r=0$.

For a candidate file $\mathcal D$, the semantic residual information matrix is
\begin{equation}
    \widehat{\mathcal I}_{\mathcal D}
    =
    \frac{1}{N_{\mathcal D}}
    \sum_{(q,p,n)\in\mathcal D}
    w\,
    rr^\top ,
    \label{eq:eci_information_matrix_theory}
\end{equation}
where $N_{\mathcal D}$ is the number of explicit negatives in the flattened
candidate file. The final diagnostic is
\begin{equation}
    \eci(\mathcal D)
    =
    \log\det\!\left(
        I_d+
        \widehat{\mathcal I}_{\mathcal D}
    \right).
    \label{eq:eci_final_theory}
\end{equation}
The log-determinant rewards non-redundant residual volume rather than raw
hardness in its nonlinear regime; Section~\ref{sec:linear_regime} reports near equality with the trace at the
default scale. With zero residuals excluded from the contribution, the trace
is the mean weighted squared residual norm.

\begin{proposition}[Structural properties of \eci{}]
\label{prop:eci_structural_properties}
For a fixed frozen encoder, query/document formatting, temperature, tokenization
rule, IDF corpus, and candidate file, assume that embeddings are unit-normalized,
IDF weights are nonnegative, $N_{\mathcal D}>0$, and residual directions use
the convention $r=0$ when $v^+=v^-$. Then:

\begin{enumerate}
    \item $\widehat{\mathcal I}_{\mathcal D}$ is positive
    semidefinite.

    \item $\eci(\mathcal D)\ge 0$, with equality if and
    only if $\widehat{\mathcal I}_{\mathcal D}=0$.

    \item The factor $\rho$ is exactly the positive probability under the
    two-document restriction of the same temperature-scaled query-document
    softmax used by MNRL.

    \item For fixed total residual weight, the log-determinant objective assigns
    strictly larger value when positive weight is spread across at least two
    non-collinear residual directions than when the same total weight is
    concentrated in a single unit direction. More generally, additional
    linearly independent weighted residual directions contribute positive
    higher-order volume terms.
\end{enumerate}
\end{proposition}

\begin{proof}
For every triplet, $\rho,\eta\in(0,1)$ because they are sigmoid
outputs. Since $T(q)\cap T(n)\subseteq T(q)$ and IDF weights are nonnegative,
$0\le C(q,n)\le 1$, so
Eq.~\eqref{eq:eci_lexical_residual_theory} gives
$\psi=1-C(q,n)\in[0,1]$ (using the convention $C(q,n)=0$ when the
denominator vanishes). Hence
\[
    w=\rho\eta\psi\ge 0 .
\]
For any vector $x\in\mathbb R^d$,
\begin{multline*}
x^\top \bigl[ w\, rr^\top \bigr] x \\
= w\, x^\top rr^\top x \\
= w\, \bigl(x^\top r\bigr)^2 \ge 0.
\end{multline*}
Each weighted outer product is positive semidefinite. Nonnegative sums preserve
positive semidefiniteness, so
$\widehat{\mathcal I}_{\mathcal D}$ is positive semidefinite.

Let $\lambda_1,\ldots,\lambda_d$ be the eigenvalues of
$\widehat{\mathcal I}_{\mathcal D}$. Since the matrix is
positive semidefinite, $\lambda_\ell\ge 0$ for all $\ell$. Therefore
\[
\begin{aligned}
\eci(\mathcal D)&=\log\det(I_d+\widehat{\mathcal I}_{\mathcal D})\\
&=\sum_{\ell=1}^d\log(1+\lambda_\ell)\ge0.
\end{aligned}
\]
with equality if and only if $\lambda_\ell=0$ for all $\ell$, which for a positive
semidefinite matrix holds if and only if
$\widehat{\mathcal I}_{\mathcal D}=0$.

The identity for $\rho$ follows directly from restricting the MNRL denominator
to the labeled positive $p$ and one explicit negative $n$:
\[
\begin{aligned}
\rho&=\frac{\exp(u^\top v^+/\tau)}
{\exp(u^\top v^+/\tau)+\exp(u^\top v^-/\tau)}\\
&=\sigma\!\left(\frac{u^\top v^+-u^\top v^-}{\tau}\right).
\end{aligned}
\]
Thus $\rho$ is the positive probability under the two-document
temperature-scaled target softmax.

For redundancy sensitivity, we argue on the weighted sum $\sum_j a_j r_j r_j^\top$;
the $1/N_{\mathcal D}$ factor in
$\widehat{\mathcal I}_{\mathcal D}$ is a common rescaling of the
$a_j$ and does not affect the comparison below. Triplets with zero weight or
with $v^+=v^-$ contribute zero outer products and drop out of the sum, so let
$r_1,\ldots,r_m$ denote the unit residual directions of the remaining triplets,
with positive weights $a_1,\ldots,a_m$ and $W=\sum_{j=1}^m a_j$. If all weight
is concentrated in one unit direction $r$, the only nonzero eigenvalue is $W$,
so
\[
    \log\det(I_d+Wrr^\top)=\log(1+W).
\]

If at least two positive-weighted directions are non-collinear, let
$R=[r_1,\ldots,r_m]$ and $A=\operatorname{diag}(a_1,\ldots,a_m)$, and write
$X=RA^{1/2}$, so that $RAR^\top=XX^\top$. Sylvester's determinant identity
$\det(I_d+XX^\top)=\det(I_m+X^\top X)$ gives
\[
    \det(I_d+RAR^\top)
    =
    \det(I_m+A^{1/2}R^\top R A^{1/2}).
\]
Let $G=R^\top R$ and $\Gamma=A^{1/2}GA^{1/2}=X^\top X$. Then $\Gamma$ is
positive semidefinite, $\Gamma_{jj}=a_j$, and
\[
    \operatorname{tr}(\Gamma)=\sum_{j=1}^m a_j=W.
\]
The determinant expansion of $I_m+\Gamma$ in principal minors gives
\[
    \det(I_m+\Gamma)
    =
    1+\operatorname{tr}(\Gamma)
    +
    \sum_{\substack{S\subseteq\{1,\ldots,m\}\\ |S|\ge 2}}
    \det(\Gamma_S),
\]
where $\Gamma_S$ is the principal submatrix indexed by $S$. All principal minors
are nonnegative. Non-collinearity gives a subset $S=\{j,k\}$ such that
\[
    \det(\Gamma_S)
    =
    a_j a_k
    \left(1-(r_j^\top r_k)^2\right)
    >
    0 .
\]
Consequently,
\[
    \det(I_d+RAR^\top)
    =
    \det(I_m+\Gamma)
    >
    1+W,
\]
and hence
\[
    \log\det(I_d+RAR^\top)>\log(1+W).
\]
Thus spreading the same total weight across non-collinear directions yields a
larger log-determinant than concentrating it in one direction. Higher-order
Gram principal minors add nonnegative volume terms and are positive for
linearly independent weighted directions.
\end{proof}

\subsection{Local Pairwise Gradient for the Hardness Comparison}
\label{sec:loss_reduction_information_volume}

We next relate the matrix form to contrastive optimization to motivate our design.
The connection is local: exact for the two-document loss at frozen embeddings,
and tied to parameter updates only through encoder linearization.

For one triplet, define the two-document loss
\begin{equation}
\begin{aligned}
    \ell_2(q,p,n)
    &=
    -\log \rho
    \\
    &=
    \log\!\left(
        1+
        \exp\!\left(
            -\frac{u^\top(v^+-v^-)}{\tau}
        \right)
    \right).
\end{aligned}
\label{eq:pairwise_loss_connection}
\end{equation}
Let $\delta=v^+-v^-$ and let
$r=\delta/\|\delta\|_2$ when $\delta\ne 0$.

\begin{proposition}[Pairwise gradient information]
\label{prop:pairwise_gradient_information}
Treat the normalized query embedding $u$ as an ambient Euclidean variable and
hold $v^+$ and $v^-$ fixed. If $\delta=v^+-v^-\ne 0$, then
\begin{equation}
    \nabla_u \ell_2(q,p,n)
    =
    -\frac{1-\rho}{\tau}\,\delta,
    \label{eq:pairwise_query_gradient}
\end{equation}
and therefore
\begin{equation}
\begin{aligned}
    &\nabla_u \ell_2(q,p,n)\nabla_u \ell_2(q,p,n)^\top\\
    &\quad=
    \frac{(1-\rho)^2}{\tau^2}
    \|v^+-v^-\|_2^2 \\
    &\quad \times\,
    rr^\top .
\end{aligned}
\label{eq:pairwise_gradient_outer_product}
\end{equation}
\end{proposition}

\begin{proof}
Let $\Delta=\frac{u^\top\delta}{\tau}$. Then $\rho=\sigma(\Delta)$ and
$\ell_2=\log(1+\exp(-\Delta))$.
Since $\frac{\partial \ell_2}{\partial \Delta} = -(1-\sigma(\Delta)) = -(1-\rho)$,
and $\nabla_u\Delta=\frac{\delta}{\tau}$, Eq.~\eqref{eq:pairwise_query_gradient}
follows by the chain rule. Taking the outer product and using $\delta=\|\delta\|_2 r$
gives Eq.~\eqref{eq:pairwise_gradient_outer_product}.
\end{proof}

\noindent\emph{Remark.} If $u$ is instead constrained to the unit sphere, the
Riemannian gradient is the tangential projection
$(I_d-uu^\top)\nabla_u \ell_2(q,p,n)$; this changes the update direction but
not its explicit $(1-\rho)$ factor. Projection changes directions and norms, so the ambient gradient remains a local surrogate.

\medskip

For fixed $\delta$, the ambient query-gradient energy scales with
$(1-\rho)^2\|\delta\|_2^2/\tau^2$. The composite instead uses
$\rho\eta\psi$ on normalized residual directions. These are distinct
weighting choices, compared empirically in Table~\ref{tab:gate_ablations};
the gradient identity does not justify either weighting as a transfer model.

\section{Robustness Checks}
\label{sec:appendix_robustness}

These checks expand Section~\ref{sec:ablation_sample_size_stability}.

\subsection{Sample-Size Stability}
\label{sec:appendix_sample_size_stability}

\eci{} rankings are unchanged at $25\%$, $50\%$, $75\%$, and $100\%$ of each
seeded sample (Table~\ref{tab:sample_size_rank_stability}). LLM remains highest
among non-hybrid sources, and Dense+LLM remains highest among hybrid sources.
The stability value of $1.000$ denotes exact agreement with the full-sample
rank order. This does not prove convergence for arbitrary datasets, but it
indicates that the observed ordering is not an artifact of using the maximum
available sample size in this experiment.

\begin{table}[!ht]
    \centering
    \small

    \begin{tabular}{lcccc}
        \toprule
        \textbf{Setting} & $25\%$ & $50\%$ & $75\%$ & $100\%$ \\
        \midrule
        Non-hybrid & 1.000 & 1.000 & 1.000 & \textbf{1.000} \\
        Hybrid & 1.000 & 1.000 & 1.000 & \textbf{1.000} \\
        \bottomrule
    \end{tabular}
\caption{Rank stability as sampled rows increase, using lexical-residual
    pairwise \eci{}. Every fraction reproduces the full-sample order, LLM $>$
    Dense $>$ BM25 for non-hybrid and Dense+LLM $>$ BM25+LLM $>$ BM25+Dense
    for hybrid; $1.000$ denotes exact agreement with it.}
\label{tab:sample_size_rank_stability}
\end{table}

\subsection{Tokenizer, IDF, and Failure-Case Sensitivity}
\label{sec:ablation_tokenizer_idf_failure}

We test the implementation-sensitive gates $\eta$ and $\psi$ on $5{,}000$
seeded rows per source. The default lexical setting uses alphanumeric tokens and
pooled-document IDF. Tokenizer variants are alphanumeric, alphanumeric length at
least two, whitespace, and character trigrams. IDF variants are pooled
documents, source-specific documents, pooled query-document text, and uniform
weights.

Table~\ref{tab:tokenizer_idf_sensitivity} shows modest score shifts but stable
rankings: LLM $>$ Dense $>$ BM25 for non-hybrid sources, and Dense+LLM $>$
BM25+LLM $>$ BM25+Dense for hybrid sources. The largest tokenizer effects occur
when changing the granularity of lexical matching, especially for dense and
LLM-containing sources. IDF-corpus changes have smaller effects for LLM and
Dense+LLM than for BM25-heavy sources. This pattern is consistent with the role
of $\psi$: BM25 negatives often share many query terms, so their residual
weight depends more strongly on the lexical coverage estimate. LLM negatives
have lower mean coverage and higher mean $\psi$, making their \eci{} scores less
tied to the exact IDF corpus.

\begin{table*}[!ht]
    \centering
    \small
    \setlength{\tabcolsep}{4pt}

    \begin{tabular}{llccccc}
        \toprule
        \textbf{Setting} & \textbf{Source} &
        \textbf{Base \eci{}} &
        \textbf{Tokenizer Range} &
        \textbf{IDF Range} &
        \textbf{Mean $C$} &
        \textbf{Mean $\psi$} \\
        \midrule
        \multirow{3}{*}{Non-hybrid}
            & BM25  & 0.0885 & 0.0882 to 0.0934 & 0.0884 to 0.1020 & 0.712 & 0.288 \\
            & Dense & 0.1086 & 0.0926 to 0.1118 & 0.1085 to 0.1131 & 0.613 & 0.387 \\
            & LLM   & \textbf{0.1860} & \textbf{0.1720 to 0.1900} & \textbf{0.1857 to 0.1866} & 0.469 & 0.531 \\
        \midrule
        \multirow{3}{*}{Hybrid}
            & BM25+Dense & 0.0986 & 0.0914 to 0.1026 & 0.0985 to 0.1076 & 0.662 & 0.338 \\
            & BM25+LLM   & 0.1373 & 0.1312 to 0.1417 & 0.1371 to 0.1438 & 0.590 & 0.410 \\
            & Dense+LLM  & \textbf{0.1478} & \textbf{0.1328 to 0.1514} & \textbf{0.1477 to 0.1497} & 0.539 & 0.461 \\
        \bottomrule
    \end{tabular}
\caption{Tokenizer and IDF-corpus sensitivity for $\psi$. Baseline uses
    alphanumeric tokenization and pooled document IDF. Ranges vary one factor at
    a time. Computed on $5{,}000$ seeded rows per source with frozen
    \texttt{prdev/mini-gte} embeddings.}
\label{tab:tokenizer_idf_sensitivity}
\end{table*}

Table~\ref{tab:eta_psi_failure_buckets} reports failure buckets. We use
$\rho<0.5$ for inversions, $\eta\le0.25$ for low locality, and $C\ge0.50$ for
high lexical coverage. Valid High-$C$ denotes $\rho\ge0.75$, $\eta\ge0.75$,
and $C\ge0.50$; Valid Low-$\eta$ denotes $\rho\ge0.75$, $\psi\ge0.75$, and
$\eta\le0.25$.

BM25-heavy sources have the highest high-coverage rates, so $\psi$ removes much lexical shortcut signal. Non-hybrid LLM negatives have lower inversion and coverage rates, consistent with their higher \eci{}. The valid high-coverage and valid low-$\eta$ buckets expose cases where gates may suppress plausible negatives, making failure modes auditable. The buckets also separate distinct failure modes. An inversion indicates that the frozen target encoder already
prefers the candidate negative over the labeled positive, which  signal label noise, a false negative, or a model-specific weakness. Low $\eta$ indicates that the negative is not close to the positive relative to the query, so it may be too easy or topically displaced. High lexical coverage indicates that the example may be solvable by query-term overlap alone. These categories do not prove that an individual negative is invalid, but they provide measurable audit groups for inspecting when $\eta$ and $\psi$ help and when they may be overly
conservative.

\begin{table}[!ht]
    \centering
    \scriptsize
    \setlength{\tabcolsep}{3pt}

    \begin{tabular}{lccccc}
        \toprule
        \textbf{Source} &
        \textbf{$\rho<0.5$} &
        \textbf{$\eta\le0.25$} &
        \textbf{$C\ge0.50$} &
        \makecell{\textbf{Valid}\\\textbf{high-$C$}} &
        \makecell{\textbf{Valid}\\\textbf{low-$\eta$}} \\
        \midrule
        BM25  & 16.8 & 43.5 & 89.6 & 7.9 & $<0.1$ \\
        Dense & 33.5 & 41.4 & 69.3 & 4.6 & 2.9 \\
        LLM   & 2.1  & 39.6 & 38.0 & 4.1 & 2.6 \\
        \midrule
        BM25+Dense & 25.2 & 42.4 & 79.5 & 6.2 & 1.5 \\
        BM25+LLM   & 9.5  & 41.6 & 63.8 & 6.0 & 1.3 \\
        Dense+LLM  & 17.7 & 40.5 & 53.4 & 4.3 & 2.8 \\
        \bottomrule
    \end{tabular}
\caption{Failure-bucket rates for $\eta$ and $\psi$, as percentages of
    explicit negatives in the $5{,}000$ row sensitivity sample, with frozen
    \texttt{prdev/mini-gte} embeddings. Upper block non-hybrid, lower block
    hybrid.}
\label{tab:eta_psi_failure_buckets}
\end{table}

\subsection{Temperature Sweep}
\label{sec:ablation_temperature_sweep}

The reported sweep evaluates
$\tau\in\{0.03,0.05,0.1\}$. Temperature changes absolute scale but not the
within-setting rank order, supporting $\tau=0.05$ as a stable default. Smaller
$\tau$ sharpens the gates and gives more weight to confident pairwise
differences; larger $\tau$ smooths those differences. The unchanged ordering
shows that the source ranking is not driven by a narrow temperature choice.

\section{LLM Hard-Negative Generation Prompt}
\label{sec:llm_hard_negative_prompt}

\begin{tcolorbox}[
    breakable,
    title={Hard-Negative Generation Prompt},
    colback=white,
    colframe=black!40,
    boxrule=0.5pt,
    arc=2pt,
    left=6pt,
    right=6pt,
    top=6pt,
    bottom=6pt
]
\small

You are an expert in \texttt{\{domain\_name\}}. You are given a user
query and a corresponding relevant document.

\medskip
\textbf{Example}

\texttt{\{positive\_example\}}

\medskip
\textbf{Task}

Generate \texttt{\{num\_hard\_negatives\}} hard-negative documents
for the given \texttt{user\_query} and
\texttt{positive\_document}.

Return a single JSON object containing the following fields:

\begin{itemize}
    \item \texttt{reasoning}: a brief description of the reasoning
    used to construct the hard negatives.
    \item \texttt{\{key\_bullets\}}
\end{itemize}

\textbf{Hard-Negative Requirements}

Each hard-negative document should:

\begin{enumerate}
    \item Be topically and lexically similar to the
    \texttt{positive\_document}, using related concepts,
    terminology, or keywords.

    \item Appear relevant to the \texttt{user\_query} at first
    glance, while differing in a subtle but important way such
    that it does not satisfy the user's information need.

    \item Be plausible, coherent, and factually reasonable as a
    standalone document.

    \item Be diverse from the other generated negatives in
    content, framing, source style, or topical emphasis.
\end{enumerate}

\textbf{Difficulty and Length}

\begin{itemize}
    \item The document should be appropriate for a reader at the
    \texttt{\{difficult\_level\}} education level.

    \item Its length should be approximately
    \texttt{\{length\}} that of the
    \texttt{positive\_document}.
\end{itemize}

\textbf{Output Format}

Return only a valid JSON object. Do not include explanations,
commentary, Markdown, or any text outside the JSON object.

\end{tcolorbox}
\end{document}